\documentclass[12pt]{amsart}
\usepackage{ amssymb,amsmath,amsfonts,amsthm,nomencl,mathrsfs} 
\usepackage[arrow, matrix, curve]{xy}
\usepackage[latin1]{inputenc}
\usepackage{a4wide}

\newcommand{\IC}{\mathbb{C}}
\newcommand{\IR}{\mathbb{R}}

\newcommand{\question}[1]{\leavevmode{\marginpar{\tiny%
$\hbox to 0mm{\hspace*{-0.5mm}$\leftarrow$\hss}%
\vcenter{\vrule depth 0.1mm height 0.1mm width \the\marginparwidth}%
\hbox to 0mm{\hss$\rightarrow$\hspace*{-0.5mm}}$\\\relax\raggedright #1}}}

\newcommand{\ILL}{\mathscr{L}}

\newcommand{\IHH}{\mathscr{H}}

\newcommand{\dom}{\mathrm{Dom}}

\newcommand{\Id}{{d}}

\newcommand{\f}{\frac}

\newcommand{\N}{\mathbb{N}}
\newcommand{\Z}{\mathbb{Z}}
\newcommand{\slim}{\mathop{\rm{st}\rule[0.5ex]{1ex}{0.1ex}\mathrm{lim}}}

\newcommand\newdot{{\kern.8pt\cdot\kern.8pt}}
\def\nbull{{\raise1.5pt\hbox{\bf .}}}

\theoremstyle{plain}            
\newtheorem{theorem}{theorem}[section]
\newtheorem{Lemma}[theorem]{Lemma}
\newtheorem{Corollary}[theorem]{Corollary}
\newtheorem{Theorem}[theorem]{Theorem}
\newtheorem{Proposition}[theorem]{Proposition}

\theoremstyle{definition}

\newtheorem{Example}[theorem]{Example}

\numberwithin{equation}{section}

\allowdisplaybreaks[1]

\newcommand{\Hmm}[1]{\leavevmode{\marginpar{\tiny%
			$\hbox to 0mm{\hspace*{-0.5mm}$\leftarrow$\hss}%
			\vcenter{\vrule depth 0.1mm height 0.1mm width \the\marginparwidth}%
			\hbox to 0mm{\hss$\rightarrow$\hspace*{-0.5mm}}$\\\relax\raggedright #1}}}

\begin{document}

\begin{titlepage}

\title{Scattering the geometry of weighted graphs}

\author[B. G\"uneysu]{Batu G\"uneysu}

\author[M Keller]{Matthias Keller}

\address{Institut für Mathematik, Humboldt-Universität zu Berlin, 12489 Berlin, Germany}
\email{gueneysu@math.hu-berlin.de}

\address{Institut für Mathematik, Universität Potsdam, 14476 Potsdam, Germany }

\end{titlepage}

\maketitle 

\begin{abstract} Given two weighted graphs $(X,b_k,m_k)$, $k=1,2$ with $b_1\sim b_2$ and $m_1\sim m_2$, we prove a weighted $L^1$-criterion for the existence and completeness of the wave operators $
W_{\pm}(H_{2},H_1, I_{1,2})$, where $H_k$ denotes the natural Laplacian in $\ell^2(X,m_k)$ w.r.t. $(X,b_k,m_k)$ and $I_{1,2}$ the trivial identification of $\ell^2(X,m_1)$ with $\ell^2(X,m_2)$. In particular, this entails a very general criterion for the absolutely continuous spectra of $H_1$ and $H_2$ to be equal.

\end{abstract}

\setcounter{page}{1}
\section{Introduction}

There is a long history on studying the stability of the absolutely continuous spectrum under various types of perturbations of a sufficiently \lq{}well-behaved\rq{} self-adjoint operator $H_0$ on a Hilbert space of the form $L^2(X,\mu)$. Classical results include perturbations of $H_0=-\Delta$ in $L^2(\IR^m)$ by potentials $V:\IR^m\to \IR$  \cite{Re3}, leading to the comparison of $H_0$ with $H=-\Delta+V$. Very sophisticated results for perturbations by potentials have been obtained in the one dimensional case $m=1$ and also in the discrete case for Jacobi matrices \cite{Deift,Killip,Kiselev,Rem} (see e.g. also \cite{LS} and the references therein for perturbations by random potentials).\\
Another well-studied class of perturbations is provided by replacing the space $X$ by some open subset $U\subset X$: here one puts some \lq\lq{}boundary conditions\rq\rq{} on $H_0$ in order to obtain the perturbed self-adjoint operator $H=H_U$ in $L^2(U)$ and compares the absolutely continuous spectra of $H_0$ and $H=H_U$. Concerning this perturbation theory (\lq{}scattering by obstacles\rq{}), we mention \cite{peter}, where so the called Grothendieck factorization principle was introduced as a tool in this context (see also \cite{Demuthetal,Demuth}).\\
A problem in perturbation which is morally rather different from the above problems has been initiated very recently: here, the idea is to assume that the free operator $H_0$ and the measure $\mu$ are defined in terms of the \lq{}geometry\rq{} of the space $X$. Then, keeping the set $X$ fixed, the inital geometry on $X$ is perturbed, leading to a perturbed operator $H$ that typically lives on a Hilbert space of the form $L^2(X,\mu\rq{})$. One then aims to study the absolutely continuous spectrum of $H$ by comparing $H$ with $H_0$. Technical problems arise here through the fact that in general one has $\mu\neq \mu\rq{}$, so that one has to compare operators living in different Hilbert spaces. A prototype of the above situation is provided by Riemannian manifolds: here $X$ is a smooth noncompact manifold which carries a \lq\lq{}nice\rq\rq{} Riemannian metric $g_0$, encoding the free geometry, and the Riemannian metric of interest $g$, encoding the perturbed geometry. One aims to study the absolutely continuous spectrum of the Laplace-Beltrami operator $H=-\Delta_{g}$ acting in $L^2(X,\mu_{g})$ by that of $H_0=-\Delta_{g_0}$ acting in $L^2(X,\mu_{g_0})$. A natural approach in this context is provided by studying the associated wave operators, which is a genuine two-Hilbert-space scattering problem \cite{salo,HPW,GuenThal,BGM,HP} in this case, as the underlying Riemannian volume measures $\mu_{g_0}$ and $\mu_{g}$ are different.\vspace{2mm}

As a contribution to the latter class of problems, in this paper, we aim to study perturbations of the geometry of weighted discrete graphs in terms of the edge and vertex weights. Apart from the work \cite{CT} on trees and perturbations of Jacobi matrices (cf. \cite{BL}  and the references therein), we are not aware of any results in this  direction, so that our results should provide the first systematic study of this topic. Moreover, the stability results we obtain include the existence of the wave operators. \\
Let us be more specific: assume we are given two weighted graphs $(X,m_1, b_1)$ and $(X, m_2,b_2)$ over the countable set $X$. Let $H_k$ denote the natural self-adjoint realization of the weighted Laplacian induced by $(X,m_k, b_k)$ in $\ell^2(X,m_k)$, defined by means of quadratic forms (cf. Section \ref{main} for the precise definitions). If we assume\footnote{Given a set $A$ and functions $\phi_1,\phi_2:A\to [0,\infty)$ we write $\phi_1\sim \phi_2$, if there exist constants $c,C>0$ such that for all $a\in A$ one has $c\phi_1(a)\leq \phi_2(a)\leq C\phi_1(a)$.} $m_1\sim m_2$, then one has $\ell^2(X,m_1)=\ell^2(X,m_2)$ and we get the trivial bounded and bijective identification operator
$$
I_{1,2}: \ell^2(X,m_1)\longrightarrow\ell^2(X,m_2), \quad f\longmapsto f.
$$
In this paper we are concerned with finding very general assumptions on the distortion of $m_2$ from $m_1$ and of $b_2$ from $b_1$ that imply the existence and completeness of the wave operators
$$
W_{\pm}(H_{2},H_1, I_{1,2})=\slim_{t\to\pm\infty}\exp(itH_{2})I_{1,2}\exp(-itH_{1})\pi_{\mathrm{ac}}(H_1).
$$
In particular, we want to allow arbitrary, possibly locally infinite weighted graphs. As the existence and completeness of the $W_{\pm}(H_{2},H_1, I_{1,2})$ implies that the absolutely continuous spectra of the $H_k$\rq{}s are equal, such a  result provides a perturbative of calculating new absolutely continuous spectra from known ones, as long as the new weighted geometry does not differ too much from the old one, in a way to be made precise below.\\
In fact, with the function
$$
m_{1,2}:X\longrightarrow (0,\infty), \quad m_{1,2}(x)=m_1(x)/m_2(x),
$$
which measures the multiplicative distortion of one vertex weight function from the other, and likewise
$$
b_{1,2}:X\times X\longrightarrow [0,\infty),\quad b_{1,2}(x,y):= \begin{cases}&b_1(x,y)/b_2(x,y),\text{ if $b_2(x,y)\ne 0$}\\
&1, \text{else},
\end{cases}
$$
a function that measures the multiplicative distortion of one edge weight function from the other, our main result reads as follows (cf. Theorem \ref{main} below):\vspace{2mm}

\emph{Assume $m_1\sim m_2$, $b_1\sim b_2$, and that there exists $s>0$ such that for both $k=1,2$ one has
\begin{align}\label{popo}
&\sum_{x\in X}\big|m_{1,2}(x)^{\frac{1}{2}}-m_{1,2}(x)^{-\frac{1}{2}}\big|\exp(-sH_k)(x,x)m_k(x) <\infty,\\
& \sum_{x\in X}\sum_{y\in X}\big|b_{1,2}(x,y)^{\frac{1}{2}}-b_{1,2}(x,y)^{-\frac{1}{2}}\big| \big(\exp(-sH_k)(x,x)+\exp(-sH_k)(y,y)\big)b_k(x,y)<\infty.
\end{align}
Then the wave operators 
$$
W_{\pm}(H_{2},H_1, I)=\slim_{t\to\pm\infty}\exp(itH_{2})I_{1,2}\exp(-itH_{1})\pi_{\mathrm{ac}}(H_1)
$$
exist and are complete. Moreover, the $W_{\pm}\big(H_{2},H_1, I_{1,2}\big)$ are partial isometries with initial space $\mathrm{Ran} \: \pi_{\mathrm{ac}}(H_1)$ and final space $\mathrm{Ran} \: \pi_{\mathrm{ac}}(H_2)$, and one has $\mathrm{spec}_{\mathrm{ac}}(H_1)=\mathrm{spec}_{\mathrm{ac}}(H_2)$. } \vspace{2mm}

Using the trivial bounds 
$$
\exp(-sH_k)(z,z) \leq 1/m_k(z)\quad\text{ for all $s\geq 0$, $ z\in X$},
$$
it follows that the assumptions (\ref{popo}) are satisfied, if for both $k$ one has
\begin{align*}
&\sum_{x\in X}\big|m_{1,2}(x)^{\frac{1}{2}}-m_{1,2}(x)^{-\frac{1}{2}}\big| <\infty, \\
& \sum_{x\in X}\sum_{y\in X}\big|b_{1,2}(x,y)^{\frac{1}{2}}-b_{1,2}(x,y)^{-\frac{1}{2}}\big|\big(m_k(x)^{-1}+m_k(y)^{-1}\big)b_k(x,y)  <\infty,
\end{align*}
a criterion which only involves geometric quantities (cf. Corollary \ref{main2} below).

\section{Main results}

In the sequel, we consider all our spaces of functions to be complex-valued. Let $X$ be a countable set, equipped with its discrete topology. Let $b:X\times X\to [0,\infty)$ be graph over $X$, that is, $b$ is a symmetric function which is zero on the diagonal and which satisfies
$$
\sum_{y\in X} b(x,y)<\infty\quad\text{ for all $y\in X$}.
$$
In other words, all $(x,y)$ with $b(x,y)>0$ are considered as being connected by an edge, weighted according to the value of $b(x,y)$. Considering $b$ as a measure on $X\times X$ in the obvious way (the edge measure), we then get the Hilbert space $\ell^2(X\times X,b)$, noting that the scalar pruduct is given by
\begin{align*}
\langle{\alpha},{\beta }\rangle_{b}=\frac{1}{2} \sum_{X\times X} \alpha \cdot\overline{\beta}\cdot b=
\frac{1}{2} \sum_{x,y\in X} \alpha(x,y) \overline{\beta(x,y)} b(x,y).
\end{align*}

 We define a nonnegative sesqui-linear form $\mathcal{Q}_b$ by means of  
\begin{align*}
    \mathcal{Q}_b:C_c(X)\times C_c(X)\longrightarrow \IC, \quad 
    \mathcal{Q}_b(f,g)=\langle{d f},{dg }\rangle_{b},
\end{align*}
with the discrete differential
\begin{align*}
    &d:C(X)\longrightarrow  C(X\times X) ,\quad df(x,y)=f(x)-f(y).
\end{align*}
Given in addition function $m:X\to (0,\infty)$, again to be considered as a measure on $X$ (the vertex measure), the triple $(X,b,m)$ is called a \emph{weighted graph}. There is the complex Hilbert space $\ell^{2}(X,m)$ of complex-valued functions on $X$ which are square summable with respect to $m$. Explicitly, the scalar product on $\ell^{2}(X,m)$ is given by
\begin{align*}
\langle{f},{g }\rangle_{m}=  \sum_{X }f\cdot\overline{g}\cdot m=
\sum_{x\in X}f(x) \overline{g(x)}m(x).
\end{align*}
Then $\mathcal{Q}_b$ is closable in $\ell^{2}(X,m)$ and its closure, a densely defined closed nonnegative sesqui-linear form in $\ell^{2}(X,m)$, is denoted with $Q_{b,m}$. For all $f ,g\in \dom(Q_{b,m})$ one has
$$
Q_{b,m}(f,g)=\langle{d f},{dg  }\rangle_{b}.
$$
We denote with $H_{b,m}$ the nonnegative self-adjoint operator in $\ell^{2}(X,m)$ corresponding to $Q_{b,m}$. Then one has $\dom(Q_{b,m})=\dom(\sqrt{H_{b,m}})$, and for all $f$, $g$ in the latter domain of definition one has
$$
Q_{b,m}(f,g)=\left\langle \sqrt{H_{b,m}}f,\sqrt{H_{b,m}}g\right\rangle_{m}.
$$

\vspace{1.2mm}

Consider now two weighted graphs $(X,b_k,m_{k})$, $k=1,2$, over $X$. We denote the corresponding scalar products, forms and operators by
$\langle{\cdot},{\cdot }\rangle_{b_k}$, $\langle{\cdot},{\cdot
}\rangle_{m_{k}}$, $Q_{k}:=Q_{b_k,m_k}$ and $H_{k}:=H_{b_k,m_k}$. If we assume $m_1\sim m_2$, then one has $\ell^2(X,m_1)=\ell^2(X,m_2)$ and there is the trivial bounded and bijective identification operator
$$
I_{1,2}: \ell^2(X,m_1)\longrightarrow\ell^2(X,m_2), \quad f\longmapsto f.
$$
One has $I_{1,2}^{-1}=I_{2,1}$, and $I_{1,2}^*f(x)=m_{1,2}(x)f(x)$, where
$$
m_{1,2}:X\longrightarrow (0,\infty), \quad m_{1,2}(x)=m_1(x)/m_2(x).
$$
Define
$$
b_{1,2}:X\times X\longrightarrow [0,\infty),\quad b_{1,2}(x,y):= \begin{cases}&b_1(x,y)/b_2(x,y),\text{ if $b_2(x,y)\ne 0$}\\
&1, \text{else}.
\end{cases}
$$
Note that $1/m_{1,2}=m_{2,1}$ and $1/b_{1,2}=b_{2,1}$, and that if $m_1\sim m_2$ and $b_1\sim b_2$, then we get
$$
 I_{1,2}(\dom(Q_1))=\dom(Q_2), \quad \langle{\cdot},{\cdot  }\rangle_{b_2}=\langle{b_{1,2}(\cdot)},{\cdot  }\rangle_{b_1}.
$$

Here comes our main result:

\begin{Theorem}\label{main} Assume $m_1\sim m_2$, $b_1\sim b_2$, and that there exists $s>0$ such that for both $k=1,2$ one has
\begin{align*}
&\sum_{x\in X}\big|m_{1,2}(x)^{\frac{1}{2}}-m_{1,2}(x)^{-\frac{1}{2}}\big|\exp(-sH_k)(x,x)m_k(x) <\infty,\\
& \sum_{x\in X}\sum_{y\in X}\big|b_{1,2}(x,y)^{\frac{1}{2}}-b_{1,2}(x,y)^{-\frac{1}{2}}\big| \big(\exp(-sH_k)(x,x)+\exp(-sH_k)(y,y)\big)b_k(x,y)<\infty.
\end{align*}
Then the wave operators 
$$
W_{\pm}(H_{2},H_1, I_{1,2})=\slim_{t\to\pm\infty}\exp(itH_{2})I_{1,2}\exp(-itH_{1})\pi_{\mathrm{ac}}(H_1)
$$
exist and are complete, where $\pi_{\mathrm{ac}}(H_k)$ denotes the projection onto the $H_k$-absolutely continuous subspace of $\ell^2(X,m_k)$. Moreover, the $W_{\pm}\big(H_{2},H_1, I_{1,2}\big)$ are partial isometries with inital space $\mathrm{Ran} \: \pi_{\mathrm{ac}}(H_1)$ and final space $\mathrm{Ran} \: \pi_{\mathrm{ac}}(H_2)$, and one has $\mathrm{spec}_{\mathrm{ac}}(H_1)=\mathrm{spec}_{\mathrm{ac}}(H_2)$.
\end{Theorem}

Noting that
$$
\exp(-sH_k)(z,z) \leq 1/m_k(z)\quad\text{ for all $s\geq 0$, $ z\in X$},
$$
a trivial consequence of $\exp(-sH_k)$ being a self-adjoint contraction in $\ell^2(X,m_k)$,
we get:

\begin{Corollary}\label{main2} Assume $m_1\sim m_2$, $b_1\sim b_2$, and that for both $k=1,2$ one has
\begin{align*}
&\sum_{x\in X}\big|m_{1,2}(x)^{\frac{1}{2}}-m_{1,2}(x)^{-\frac{1}{2}}\big| <\infty, \\
& \sum_{x\in X}\sum_{y\in X}\big|b_{1,2}(x,y)^{\frac{1}{2}}-b_{1,2}(x,y)^{-\frac{1}{2}}\big|\big(m_k(x)^{-1}+m_k(y)^{-1}\big)b_k(x,y)  <\infty.
\end{align*}
Then the wave operators 
$$
W_{\pm}(H_{2},H_1, I_{1,2})=\slim_{t\to\pm\infty}\exp(itH_{2})I_{1,2}\exp(-itH_{1})\pi_{\mathrm{ac}}(H_1)
$$
exist and are complete. Moreover, the $W_{\pm}\big(H_{2},H_1, I_{1,2}\big)$ are partial isometries with inital space $\mathrm{Ran} \: \pi_{\mathrm{ac}}(H_1)$ and final space $\mathrm{Ran} \: \pi_{\mathrm{ac}}(H_2)$, and one has $\mathrm{spec}_{\mathrm{ac}}(H_1)=\mathrm{spec}_{\mathrm{ac}}(H_2)$.
\end{Corollary}

\section{Applications of Corollary \ref{main2}}

Recall that a weighted graph $ (X,b,m) $ is said to have \emph{standard weights}, if $ b  $ takes values in $ \{0,1\} $ and $ m $ is constantly $ 1 $. Observe, that graphs with standard weights are always locally finite. Furthermore, let the uniform closure of the functions with finite support on a discrete set $ Y $ be denoted by $ C_{0}(Y) $.

\begin{Corollary}\label{main3} Let $ (X,b_{1},m_{1}) $ be a graph with standard weights and $ \beta\in C_{0}(X\times X) $  symmetric, zero on the diagonal and $ \mu\in C_{0}(X) $ be such that with
	\begin{align*}
	b_{2}:=b_{1}+\beta\quad\mbox{and}\quad m_{2}:=m_{1}+\mu,
	\end{align*}
$b_2$ is nonnegative and  $m_2$ is strictly positive. Let, furthermore, $ S_{n} \subseteq X$, $ n\in \N_{0} $, be finite pairwise disjoint subsets such that $ X=\bigcup_{n}S_{n} $ and define the averages over $ S_{n} $ to be
	\begin{align*}
	\beta_n:= \frac{1}{\# S_{n}}
	\sum_{x\in S_{n},y\in X}|\beta(x,y)|\quad\mbox{and}\quad
	\mu_n:= \frac{1}{\# S_{n}}
	\sum_{x\in S_{n}}|\mu(x)|,\qquad n\in \N_{0}.
	\end{align*}
	Assume	
	\begin{align*}
	\sum_{n\in \N_{0}}\#S_{n}\beta_{n} <\infty \quad\mbox{and}\quad
	\sum_{n\in \N_{0}}\#S_{n}\mu_{n} <\infty.
	\end{align*}
	Then the wave operators 
	$$
	W_{\pm}(H_{2},H_1, I_{1,2})=\slim_{t\to\pm\infty}\exp(itH_{2})I_{1,2}\exp(-itH_{1})\pi_{\mathrm{ac}}(H_1)
	$$
	exist and are complete. Moreover, the $W_{\pm}\big(H_{2},H_1, I_{1,2}\big)$ are partial isometries with initial space $\mathrm{Ran} \: \pi_{\mathrm{ac}}(H_1)$ and final space $\mathrm{Ran} \: \pi_{\mathrm{ac}}(H_2)$, and one has $\mathrm{spec}_{\mathrm{ac}}(H_1)=\mathrm{spec}_{\mathrm{ac}}(H_2)$.
\end{Corollary}
\begin{proof}
	Since $ \beta $ and $ \mu $ are  $ C_{0} $-functions, there are only finitely many values larger than $ 1/2 $ in absolute value. So assume w.l.o.g. $  |\beta|,|\mu |\leq 1/2 $.
	We calculate using the Taylor expansion of the square root  
	\begin{align*}
	b_{1,2}^{\frac{1}{2}}-b_{1,2}^{-\frac{1}{2}}&=\frac{(1+\beta)^{1/2}-1}{(1+\beta)^{1/2}}=\frac{1-\beta^{2}-(1+\beta)^{1/2}}{1-\beta^{2}}=-\frac{\beta}{2(1-\beta^{2})}+O(\beta^{2})\\
	m_{1,2}^{\frac{1}{2}}-m_{1,2}^{-\frac{1}{2}}&=-\frac{\mu}{2(1-\mu^{2})}+O(\mu^{2}),
	\end{align*}
where we use the Landau symbol $ O $ for real valued functions on a discrete sets $ Y $ with respect to the limit to the point  $ \infty_Y $ in the Alexandrov compactification $Y\cup \{\infty_Y\} $ of $ Y $. Hence,
	\begin{align*}
	\sum_{x\in X}\big|m_{1,2}(x)^{\frac{1}{2}}-m_{1,2}(x)^{-\frac{1}{2}}\big| \leq \sum_{n\in\N_0}\# S_{n}(\mu_{n}+O(\mu_{n}^{2}))<\infty.
	\end{align*}
 Since we assumed  $  |\beta|,|\mu |\leq 1/2 $, we infer $ m_{k}^{-1},b_{k}\leq 2 $, $ k=1,2 $	 and, therefore,
 \begin{align*}
 \sum_{x,y\in X}\big|b_{1,2}(x,y)^{\frac{1}{2}}-b_{1,2}(x,y)^{-\frac{1}{2}}\big|\big(m_k(x)^{-1}+m_k(y)^{-1}\big)b_k(x,y) \leq  \sum_{n\in\N_0}\# S_{n}(\beta_{n}+O(\beta_{n}^{2}))<\infty.
 \end{align*}
	Thus, we deduce the statement by Corollary~\ref{main2}.
\end{proof}
Next, we discuss some particular examples where we can apply the corollary above directly. 

\begin{Example}[The Euclidean lattice] Consider the Laplacian $H_1$ with standard weights on  $ \Z^{d} $. It is folklore that
	the spectrum of this operator is purely absolutely continuous 
	\begin{align*}
	\mathrm{spec}(H_1)=	\mathrm{spec}_{\mathrm{ac}}(H_1)=[0,4d],
	\end{align*} 
	which can be proven using the discrete Fourier transform.	
	Denote the distance spheres with respect to the supremum norm on $ \Z^{d} $ by $ S_{n} $, $ n\in\N_{0} $. Then,
	\begin{align*}
	\# S_{n}\sim n^{d-1},	
	\end{align*} 
	and all $ \beta\in C_{0}(\Z^{d}\times \Z^{d}) $ and $ \mu\in C_{0}(\Z^{d}) $ with 
	\begin{align*}
	\beta_{n}, \mu_n\in O(n^{-(d+\varepsilon)})
	\end{align*}
	for some $ \varepsilon>0$ satisfy the assumptions of Corollary~\ref{main3}, in particular, for those $\beta, \mu$ one has
\begin{align*}
	\mathrm{spec}(H_2)=	\mathrm{spec}_{\mathrm{ac}}(H_1)=[0,4d].\end{align*} 

For Cayley graphs of abelian groups it is also known that the Laplacian has purely continuous spectrum, \cite{higuchi}. These graphs have polynomial growth and therefore every $ \beta,\mu $ that decays a bit more than one order of the polynomial growth satisfies the assumptions of Corollary~\ref{main3}, showing that the correspoding weighted Laplacians still have a purely continuous spectrum.
\end{Example}

\begin{Example}[Regular trees] Consider the Laplacian $ H_1$ with standard weights on a $ k $-regular tree. The spectrum of this operator is purely absolutely continuous 
	\begin{align*}
	\mathrm{spec}(H_1)=	\mathrm{spec}_{\mathrm{ac}}(H)=[2(k+1)-2\sqrt{k},2(k+1)+2\sqrt{k}].
	\end{align*}  
	This is also folklore and attributed to Furstenberg. It can be seen using the recursion formula of the Green's function to compute the spectral density, see \cite{Klein}.
	Denote the distance spheres with respect to the combinatorial graph distance by $ S_{n} $, $ n\in\N $. Then,$\# S_{n}=k^{n}$, and we may allow for 
	\begin{align*}
	\beta_{n},\mu_{n}\in O(n^{-(1+\varepsilon)}k^{-n})
	\end{align*}
	for some $ \varepsilon>0 $ 
    to apply Corollary~\ref{main3}.
\end{Example}

\begin{Example}[Trees of finite cone type and Galton Watson trees] A tree of finite cone type, \cite{NW}, is generated by substitution matrix $ M\in \N_{0}^{N\times N} $ for some natural number $ N $. Starting with a root for which one has to pick a label from $ \{1,\ldots,N\} $ every vertex with label $ k $ has $ M_{k,l} $ forward neighbors of label $ l $. The Laplacian on these trees has finitely many intervals of  pure absolutely continuous spectrum provided $ M $ is irreducible and has positive diagonal, \cite{KLW,KLW2}.
Denote again the distance spheres with respect to the combinatorial graph distance  by $ S_{n} $, $ n\in\N_{0} $. Then, by elementary considerations one has $	\# S_{n} \sim \gamma^{n}$, where $ \gamma  $ is the largest eigenvalue of $ M $. Hence, we may allow for 
	\begin{align*}
	\beta_{n},\mu_{n}\in O(n^{-(1+\varepsilon)}\gamma^{-n})
	\end{align*}
	for some $ \varepsilon>0 $ 
	to apply Corollary~\ref{main3}.\\
 These considerations may be generalized to certain Galton-Watson trees. In \cite{Keller} it is shown that if the distribution of such a tree is close to a deterministic tree of finite cone type, then the Laplacians shares most of the absolutely continuous spectrum of the deterministic tree. By the arguments above we may perturb the weights on these random trees according to the allowed perturbations for the corresponding tree of finite cone type to ensure stability of the absolutely continuous spectrum.
\end{Example}

\section{Proof of Theorem \ref{main}}

Assume we are given two weighted graphs $(X,b_k,m_{k})$, $k=1,2$, over $X$. With
$$
\tilde{m}_{1,2}:=\sqrt{m_{1,2}}-1/\sqrt{m_{1,2}},\quad \tilde{b}_{1,2}:=\sqrt{b_{1,2}}-1/\sqrt{b_{1,2}},
$$

we define maximally defined multiplication operators

\begin{align*}
&S_{1,2;k}:\ell^2(X,m_k)\longrightarrow \ell^2(X,m_k),\quad S_{1,2;k}f(x)=|\tilde{m}_{1,2}(x)|^{\f{1}{2}}f(x),\\
&\widehat{S_{1,2;k}}:\ell^2(X\times X, b_k )\longrightarrow \ell^2(X\times X, b_k ),\quad  \widehat{S_{1,2;k}}\alpha(x,y)=|\tilde{b}_{1,2}(x,y)|^{\f{1}{2}}\alpha(x,y),\\
&U_{1,2}:\ell^2(X,m_1)\longrightarrow \ell^2(X,m_2),\quad  U_{1,2}f(x)= \mathrm{sgn}\big(\tilde{m}_{1,2}(x)\big)m_{1,2}(x)^{-\f{1}{2}} f(x),\\
&\widehat{U_{1,2}}:\ell^2(X\times X, b_1 )\longrightarrow\ell^2(X\times X, b_2 ),\quad \widehat{U_{1,2}}\alpha(x,y)= \mathrm{sgn}\big(\tilde{b}_{1,2}(x,y)\big) b_{1,2}(x,y)^{-\f{1}{2}}\alpha(x,y). 
\end{align*}

The operators $U_{1,2}$, $\widehat{U_{1,2}}$ are always unitary, and the operators $S_{1,2;k}$, $\widehat{S_{1,2;k}}$ are always self-adjoint and in addition bounded if $m_1\sim m_2$ and $b_1\sim b_2$. Denote with $d_k$ the closure of the operator $d$ considered as an operator from $\ell^2(X,m_k)$ to $ \ell^2 (X\times X,b_k)$, defined initially on finitely supported functions. Then one has $H_k= d_k^* d_k$, with $\dom(d_k)=\dom(\sqrt{H_k})=\dom(Q_k)$, by Proposition \ref{help}. We denote the heat semigroups with
$$
P^{(k)}_s:=\exp(-sH_k)\in \ILL(\ell^2(X,m_k)), \quad s>0,
$$ 
and we further define
$$
\hat{P}^{(k)}_s:=d_kP^{(k)}_s\in \ILL\big(\ell^2(X,m_k),\ell^2(X\times X, b_k)\big), \quad s>0.
$$
In view of $\hat{P}^{(k)}_s=d_k \exp(-sd_k^* d_k )$, the asserted boundedness of $\hat{P}^{(k)}_s$ follows again from Proposition \ref{help}.\\
The main idea for the proof of Theorem \ref{main} is to use the above constructions in order to write the operator $T$ from the abstract Belopol\rq{}skii-Birman Theorem (cf. Section \ref{wave}) essentially as a sum of operators of the form 
$$
(A_1\hat{P}^{(2)}_s)^*X_1A_2\hat{P}^{(1)}_sX_2, \quad (A_3P^{(2)}_s)^*X_3A_4P^{(1)}_sX_4,
$$
 where the $X_j$ are bounded and where the multiplication operators $A_j$ reflect the deviations $m_{1,2}$ and $b_{1,2}$ and will be chosen in a way that 
$$
(A_1\hat{P}^{(2)}_s)^*X_1A_2\hat{P}^{(1)}_sX_2, \quad (A_3P^{(2)}_s)^*X_3A_4P^{(1)}_sX_4
$$
 will be Hilbert-Schmidt. This idea stems from the scattering theory of Riemannian manifolds \cite{HPW, GuenThal, BGM} and has been called \emph{Hempel-Post-Weder formula} in \cite{GuenThal}:

\begin{Lemma}[Discrete HPW formula]\label{HPW} Assume $m_1\sim m_2$, $b_1\sim b_2$, $s>0$. Then defining the bounded operator $T_{1,2,s}:\ell^2(X,m_1)\to \ell^2(X,m_2)$ by
\begin{align*}
T_{1,2,s}:=\big(\widehat{S_{1,2;2}}\widehat{P}^{(2)}_s\big)^{*}\widehat{U_{1,2}}\widehat{S_{1,2;1}}\widehat{P}^{(1)}_s-\big(S_{1,2;2} P^{(2)}_s\big)^{*}U_{1,2}S_{1,2;1} P^{(1)}_{s/2}H_1 P^{(1)}_{s/2},
\end{align*}
the following formula holds for all $f_2\in \dom (H_2)$, $f_1\in \dom (H_2)$,
\begin{align*}
\left\langle f_2 ,T_{1,2,s}f_1\right\rangle_{m_2} =\left\langle H_2f_2, P^{(2)}_s I_{1,2} P^{(1)}_sf_1\right\rangle_{m_2} -\left\langle f_2,P^{(2)}_s I_{1,2}P^{(1)}_s H_1f_1\right\rangle_{m_2}.
\end{align*}
\end{Lemma}
\begin{proof} Adding and subtracting the term $\left\langle I_{2,1}  f_2, H_1 P^{(1)}_sf_1\right\rangle_{m_1}$ we get
\begin{align*}
&\left\langle H_2f_2, P^{(2)}_s I_{1,2} P^{(1)}_sf_1\right\rangle_{m_2} -\left\langle f_2, P^{(2)}_s I_{1,2} P^{(1)}_s H_1f_1\right\rangle_{m_2}\\
&=\left\langle H_2P^{(2)}_sf_2,  I_{1,2} P^{(1)}_sf_1\right\rangle_{m_2} -\left\langle I_{2,1}P^{(2)}_sf_2,   H_1 P^{(1)}_sf_1\right\rangle_{m_1}\\
&\quad- \left\langle  P^{(2)}_s f_2, \big(I_{1,2}-I_{2,1}^*\big) H_1 P^{(1)}_sf_1\right\rangle_{m_1}\\
&=\left\langle \Id_2P^{(2)}_sf_2,  \Id_2 I_{1,2} P^{(1)}_sf_1\right\rangle_{b_2} -\left\langle \Id_1 I_{2,1}P^{(2)}_sf_2,   \Id_1 P^{(1)}_sf_1\right\rangle_{b_1}\\
&\quad- \left\langle  P^{(2)}_s f_2, \big(I_{1,2}-I_{2,1}^*\big) H_1 P^{(1)}_sf_1\right\rangle_{m_2}.
\end{align*}
Using $\langle{\cdot},{\cdot  }\rangle_{b_1}=\langle{b_{1,2}^{-1}(\cdot)},{\cdot  }\rangle_{b_1}$ and
$$
1-b_{1,2}^{-1}=b_{1,2}^{-\frac{1}{2}}(b_{1,2}^{\frac{1}{2}}-b_{1,2}^{-\frac{1}{2}})=b_{1,2}^{-\frac{1}{2}}\tilde{b}_{1,2}=b_{1,2}^{-\frac{1}{2}}|\tilde{b}_{1,2}|\mathrm{sgn}(\tilde{b}_{1,2}),
$$
the first two terms combine to
\begin{align*}
&\left\langle \Id_2P^{(2)}_sf_2,  \Id_2 I_{1,2} P^{(1)}_sf_1\right\rangle_{b_2} -\left\langle \Id_1 I_{2,1}P^{(2)}_sf_2,   \Id_1 P^{(1)}_sf_1\right\rangle_{b_1}\\
&=\sum_{X\times X} \Id P^{(2)}_sf_2 \cdot  \overline{\Id    P^{(1)}_sf_1} \cdot b_{2} \cdot(1-b_{1,2}^{-1}) \\
&=\sum_{X\times X}    \Id P^{(2)}_sf_2 \cdot \overline{|\tilde{b}_{1,2} |^{\f{1}{2}}\cdot\mathrm{sgn}\big(\tilde{b}_{1,2} \big) b_{1,2}^{-\f{1}{2}}\cdot|\tilde{b}_{1,2} |^{\f{1}{2}} \cdot\Id  P^{(1)}_sf_1} \cdot b_{2}\\
&=\left\langle  \widehat{P}^{(2)}_sf_2 \ , \ \widehat{S}_{1,2;2} \widehat{U_{1,2}} \widehat{S}_{1,2;1} \widehat{P}^{(1)}_s f_1   \right\rangle_{b_2}\\
&=\left\langle f_2 \ , \ (\widehat{P}^{(2)}_s)^* \widehat{S}_{1,2;2} \widehat{U}_{1,2} \widehat{S}_{1,2;1} \widehat{P}^{(1)}_s f_1   \right\rangle_{b_2}.
\end{align*}
Likewise, one has
\begin{align*}
&\left\langle  P^{(2)}_s f_2, \big(I_{1,2}-I_{2,1}^*\big) H_1 P^{(1)}_sf_1\right\rangle_{m_2}\\
&=\sum_X P^{(2)}_s f_2 \cdot  \big(1-m_{1,2}^{-1}\big) \cdot\overline{H_1 P^{(1)}_sf_1}\cdot m_2\\
&=\sum_X P^{(2)}_s f_2   \cdot\overline{|\tilde{m}_{1,2}|^{\f{1}{2}}\cdot\mathrm{sgn}\big(\tilde{m}_{1,2}\big)m_{1,2}^{-\f{1}{2}} \cdot|\tilde{m}_{1,2}|^{\f{1}{2}}H_1 P^{(1)}_sf_1}\cdot m_2\\
&=\left\langle   f_2,    P^{(2)}_s  S_{1,2;2} U_{1,2}S_{1,2;1}   H_2 P^{(1)}_{s} f_1    \right\rangle_{m_2}\\
&=\left\langle   f_2,    P^{(2)}_s  S_{1,2;2} U_{1,2}S_{1,2;1} P^{(1)}_{s/2} H_2 P^{(1)}_{s/2} f_1    \right\rangle_{m_2},
\end{align*}
completing the proof.
\end{proof}

\begin{Lemma}\label{gradient} For every $s>0$ the operator 
$$
\widehat{P}^{(k)}_s=d_kP^{(k)}_s\in \ILL\big(\ell^2(X,m_k),\ell^2(X\times X, b_k)\big)
$$
is an integral operator having a uniquely determined integral kernel 
$$
\hat{P}^{(k)}_s(\cdot,\cdot,\cdot):X\times X\times X\longrightarrow \IC
$$
which for all $f\in \ell^2(X,m_k)$, $x,y\in X$ satisfies
\begin{align*}
&\hat{P}^{(k)}_sf(x,y)=\sum_{z\in X}\hat{P}^{(k)}_s(x,y,z) f(z)m_k(z),\\
 &\sum_{z\in X}|\hat{P}^{(k)}_s(x,y,z)|^2 m_k(z)\leq 2P^{(k)}_{2s}(x,x)+2P^{(k)}_{2s}(y,y).
\end{align*}
\end{Lemma}

\begin{proof} By Riesz-Fischer's duality theorem it suffices to prove that for all $f\in \ell^2(X,m_k)$ and all $x,y\in X$ one has 
$$
|\hat{P}^{(k)}_sf(x,y)|\leq ( P^{(k)}_{2s}(x,x)^{1/2}+P^{(k)}_{2s}(y,y)^{1/2})\left\|f\right\|_{m_k},
$$
showing that for fixed $s,x,y$ the linear functional $f\mapsto \hat{P}^{(k)}_sf(x,y)$ on $\ell^2(X,m_k)$ is bounded with bound 
$$
\leq  P^{(k)}_{2s}(x,x)^{1/2}+P^{(k)}_{2s}(y,y)^{1/2} .
$$ 
The latter estimate follows from the trivial bound
$$
|\hat{P}^{(k)}_sf(x,y)|= |dP^{(k)}_sf(x,y)|\leq |P^{(k)}_sf(x)|+|P^{(k)}_sf(y)|
$$
and 
\begin{align*}
&|P^{(k)}_sf(z)| \leq \sum_{z\rq{}\in X}P^{(k)}_s(z,z\rq{})|f(z\rq{})|m_k(z\rq{}) \\
&\leq \left(\sum_{z\rq{}\in X}P^{(k)}_s(z,z\rq{})^2m_k(z\rq{})\right)^{1/2}\left\|f\right\|_{m_k}=\sqrt{P^{(k)}_{2s}(z,z)}\left\|f\right\|_{m_k}
\end{align*}
by the semigroup property.
\end{proof}

Using Belopol\rq{}skii-Birman\rq{}s theorem (cf. section \ref{wave} below) we can now give the

\begin{proof}[Proof of Theorem \ref{main}] Firstly, as we have noted $m_1\sim m_2$ and $b_1\sim b_2$ implies 
$$
I_{1,2}(\dom(Q_1))=\dom(Q_2).
$$ 
The operator 
$$
A:=\big(I_{1,2}^*I_{1,2}-\mathrm{id}_{\ell^2(X,m_1)}\big)P^{(1)}_s
$$
is Hilbert-Schmidt, and thus compact: As the function $r\mapsto P^{(1)}_{r}(x,x)$ is monotonously decreasing, this follows from
$$
A(x,y)=(m_{1,2}(x)-1)P^{(1)}_s(x,y), \quad \sum_{y\in X}P^{(1)}_s(x,y)^2m_1(y) \leq P^{(1)}_{2s}(x,x),
$$
and thus 
\begin{align*}
&\sum_{x\in X}\sum_{y\in Y} |A(x,y)|^2 m_1(y) m_1(x)\leq \sum_{x\in X}   |m_{1,2}(x)-1|^2 P^{(1)}_{s/2 }(x,x)m_1(x)\\
&\leq  C \sum_{x\in X}   |m_{1,2}(x)^{\f{1}{2}}-m_{1,2}(x)^{-\f{1}{2}}|P^{(1)}_{s/2 }(x,x)m_1(x)<\infty,
\end{align*}
noting that by $m_1\sim m_{2}$ and $1/m_{1,2}=m_{2,1}$ the functions $m_{1,2}$ and $1/m_{1,2}$ are bounded.\\
The proof is finished, once we have shown that the operator $$ T_{1,2,s}=\big(\widehat{S_{1,2;2}}\widehat{P}^{(2)}_s\big)^{*}\widehat{U_{1,2}}\widehat{S_{1,2;1}}\widehat{P}^{(1)}_s-\big(S_{1,2;2} P^{(2)}_s\big)^{*}U_{1,2}S_{1,2;1} P^{(1)}_{s/2}H_1 P^{(1)}_{s/2}, $$ from Lemma \ref{HPW} is trace class. To see this, note that the product of Hilbert-Schmidt operators is trace class, that the product of a bounded operator and a Hilbert-Schmidt operator is Hilbert-Schmidt, and that adjoints of Hilbert-Schmidt operators also have this property. Thus, since $ U_{1,2} $, $ \widehat U_{1,2} $ are unitary and $ H_{1}P_{s/2}^{(1)} $ is bounded, it suffices to prove that the operators
$$
B_k:=\widehat{S_{1,2;k}}\widehat{P}^{(k)}_s,\quad k=1,2, 
$$
and 
$$
 C_1:= S_{1,2;1}P^{(1)}_{s/2},\quad C_2:= S_{1,2;2}P^{(2)}_{s}
$$
are Hilbert-Schmidt. For the former, this follows from 
\begin{align*}
&B_k(x,y,z)=|b_{1,2}(x)^{\frac{1}{2}}(x,y)-b_{1,2}(x,y)^{-\frac{1}{2}}|^{\frac{1}{2}}\widehat{P}^{(k)}_s(x,y,z),\\
&\sum_{z\in X}|\widehat{P}^{(k)}_s(x,y,z)|^2 m_k(z)\leq 2P^{(k)}_{2s}(x,x)+2P^{(k)}_{2s}(y,y),
\end{align*}
where we have used Lemma \ref{gradient}, and so by the assumptions of the theorem
\begin{align*}
&\sum_{x\in X}\sum_{y\in X}\sum_{z\in X} |B_k(x,y,z)|^2   m_k(z) b_k(x,y)\\
&\leq 2\sum_{x\in X}\sum_{y\in X}\sum_{z\in X} \big|b_{1,2}(x)^{\frac{1}{2}}(x,y)-b_{1,2}(x,y)^{-\frac{1}{2}}\big|  \big( P^{(k)}_{s/2}(x,x)+ P^{(k)}_{s/2}(y,y)\big)  b_k(x,y)
<\infty.
\end{align*}

The Hilbert-Schmidt property of $C_2$ follows from
\begin{align*}
&C_2(x,y)=|m_{1,2}(x)^{\frac{1}{2}}-m_{1,2}(x)^{-\frac{1}{2}}|^{\frac{1}{2}}P^{(2)}_s(x,y) ,\qquad  \sum_{y\in X}P^{(2)}_s(x,y)m_2(y) \leq  P^{(2)}_{2s}(x,x) ,
\end{align*}
and so by assumptions of the theorem
\begin{align*}
\sum_{x\in X}\sum_{y\in X}|C_2(x,y)|^2 m_2(x)m_2(y)\leq \sum_{x\in X}|m_{1,2}(x)^{\frac{1}{2}}-m_{1,2}(x)^{-\frac{1}{2}}|P^{(2)}_{s/2}(x,x) m_2(x)<\infty.
\end{align*}
Likewise one finds
\begin{align*}
\sum_{x\in X}\sum_{y\in X}|C_1(x,y)|^2 m_1(x)m_1(y)&\leq \sum_{x\in X}|m_{1,2}(x)^{\frac{1}{2}}-m_{1,2}(x)^{-\frac{1}{2}}|P^{(1)}_{s/2}(x,x)m_1(x) <\infty,
\end{align*}
finishing the proof of Theorem \ref{main}.
\end{proof}

\appendix 
\section{Belopol\rq{}skii-Birman theorem}\label{wave}


\begin{Theorem}\emph{(Belopol\rq{}skii-Birman)} For $k=1,2$, let $H_k\geq 0$ be a self-adjoint operator in a complex Hilbert space $\IHH_k$, where in the sequel $\pi_{\mathrm{ac}}(H_k)$ denotes the projection onto the $H_k$-absolutely continuous subspace of $\IHH_k$. Assume that $I\in\ILL(\IHH_1, \IHH_2)$ is such that the following assumptions hold:
\begin{itemize}
\item $I$ has a two-sided bounded inverse
\item One has either $I(\dom(\sqrt{H_1}))=\dom(\sqrt{H_2})$ or $I(\dom(H_1))=\dom(H_2)$.
\item The operator
\begin{align*}
(I^*I-\mathrm{id}_{\IHH_1})\exp(-r H_1):\IHH_1\to\IHH_1 \>\>\text{ is compact for some $r>0$}.
\end{align*}

\item There exists a trace class operator $T:\IHH_1\to \IHH_2$ and a number $s>0$ such that for all $f_2\in\dom(H_2)$, $f_1\in\dom(H_1)$ one has
\begin{align*}
\left\langle f_2 ,Tf_1\right\rangle_{\IHH_2}\>=\>
&\left\langle H_2f_2, \exp(-sH_{2}) I \exp(-sH_{1})f_1\right\rangle_{\IHH_2} \\
&-\left\langle f_2, \exp(-sH_{2}) I \exp(-sH_{1}) H_1f_1\right\rangle_{\IHH_2}. 
\end{align*}

\end{itemize}
Then the wave operators 
$$
W_{\pm}(H_{2},H_1, I)=\slim_{t\to\pm\infty}\exp(itH_{2})I\exp(-itH_{1})\pi_{\mathrm{ac}}(H_1)
$$
exist\footnote{$\slim_{t\to\pm\infty}$ stands for the strong limit.} and are complete, where completeness means that 
$$
\left(\mathrm{Ker} \: W_{\pm}(H_{2},H_1, I)\right)^{\perp}=\mathrm{Ran}\:  \pi_{\mathrm{ac}}(H_1), \quad\overline{\mathrm{Ran} \: W_{\pm}(H_{2},H_1, I)}=\mathrm{Ran}\:  \pi_{\mathrm{ac}}(H_2).
$$
Moreover, $W_{\pm}\big(H_{2},H_1, I\big)$ are partial isometries with inital space $\mathrm{Ran} \: \pi_{\mathrm{ac}}(H_1)$ and final space $\mathrm{Ran} \: \pi_{\mathrm{ac}}(H_2)$, and one has $\mathrm{spec}_{\mathrm{ac}}(H_1)=\mathrm{spec}_{\mathrm{ac}}(H_2)$.
\end{Theorem}

\begin{proof} This result is essentially included in Theorem XI.13 from \cite{Re3}. A detailed proof is given in \cite{GuenThal}.
\end{proof}

\section{Some facts on quadratic forms in Hilbert spaces}

The following result is certainly well-known. As we have not been able to find a precise reference, we have included a quick proof: 

\begin{Proposition}\label{help} Let $D$ be a densely defined, closed operator from a Hilbert space $\IHH$ to another Hilbert space $\widetilde{\IHH}$. Then the following assertions hold:\\
a) The nonnegative, densely defined sesquilinear form $Q_D$ in $\IHH$ given by
$$
Q_D(f_1,f_2):= \left\|Df_1\right\|^2, \quad \dom (Q_D)= \dom (D),
$$
is closed, and its associated nonnegative self-adjoint operator is $D^*D$.\\
b) For all $t>0$ the operator $D\exp(-t D^*D)$ from $\IHH$ to $\widetilde{\IHH}$ is in $\ILL(\IHH,\widetilde{\IHH})$.
\end{Proposition}

\begin{proof} a) It is checked easily that $Q_D$ is closed. Let $H_D\geq 0$ denote its associated self-adjoint operator. If $f_1\in \dom (H_D)$, then we have $f_1\in \dom (Q_D)=\dom(D)$, and for all $f_2\in \dom(D)$,
$$
\left\langle H_Df_1, f_2\right\rangle=Q_D(f_1,f_2)=\left\langle Df_1, Df_2\right\rangle,
$$
which implies $Df_1\in \dom (D^*)$ and $D^*Df_1=H_D f_1$. Conversely, if $f_1\in \dom(D)= \dom (Q_D)$ with $Df_1\in \dom(D^*)$, then for all $f_2\in \dom(D)= \dom (Q_D)$ we have 
$$
Q_D(f_1,f_2)=\left\langle Df_1, Df_2\right\rangle = \left\langle D^* Df_1, f_2\right\rangle,
$$
which implies $f_1\in \dom (H_D)$ and $H_Df_1=D^*Df_1$.\\
b) Set $H_D:=D^*D$. The polar decomposition of $D$ reads $D= U\sqrt{H_D}$, where $U$ is an everywhere defined operator from $\IHH$ to $\widetilde{\IHH}$ which maps 
$$
\overline{\mathrm{Ran}(\sqrt{H_D})}\longrightarrow \overline{\mathrm{Ran}(D)}\quad\text{  isometrically}. 
$$
Thus, we have 
$$
\left\|D\exp(-t D^*D)\right\|=\left\|U\sqrt{H_D} \exp(-t H_D)\right\|=\left\|\sqrt{H_D} \exp(-t H_D)\right\|,
$$
which is $<\infty$ for all $t>0$ by the spectral calculus.
\end{proof}

\end{document}